\documentclass[11pt]{article}

\usepackage{amssymb}
\usepackage{graphicx}

\usepackage{times}

\newenvironment{proofsketch}{\noindent {\bf Proof sketch.}\ }{\par\vskip 4mm\par}

\usepackage{xspace}
\makeatletter
\xspaceaddexceptions{\check@icr}
\makeatother

\usepackage[subrefformat=subparens]{subfig}

\usepackage{enumerate}
\usepackage{graphicx}

\usepackage{amsfonts}
\usepackage{amsmath}
\usepackage{amsthm}
\usepackage{thmtools}
\usepackage{thm-restate}

\graphicspath{{./img/}}

\newtheorem{theorem}{Theorem}
\newtheorem{lemma}[theorem]{Lemma}

\theoremstyle{definition}
\newtheorem{definition}[theorem]{Definition}
\theoremstyle{remark}

\theoremstyle{plain}
\newtheorem{observation}[theorem]{Observation}

\newcommand{\Wlog}{W.l.o.g.\xspace}
\newcommand{\Wlogs}{w.l.o.g.\xspace} 
\newcommand{\wrt}{w.r.t.\xspace}

\def\max{\mathrm{max}}

\newcommand{\addcost}{exclusive cost\xspace}
\newcommand{\Addcost}{Exclusive cost\xspace}
\newcommand{\addcosts}{exclusive costs\xspace}

\newcommand{\vaddcost}{amortized cost\xspace}
\newcommand{\Vaddcost}{Amortized cost\xspace}
\newcommand{\vaddcosts}{amortized costs\xspace}

\newcommand{\spg}{series-parallel graph\xspace}
\newcommand{\spgs}{series-parallel graphs\xspace}

\newcommand{\Spg}{Series-parallel graph\xspace}
\newcommand{\spgalg}{SPGAA\xspace}
\newcommand{\Spgalg}{SPGAA\xspace}
\newcommand{\spgalglong}{Series-Parallel Graph Arrangement Algorithm\xspace}

\newcommand{\ttg}{two-terminal graph\xspace}
\newcommand{\ttgs}{two-terminal graphs\xspace}
\newcommand{\Ttg}{Two-terminal graph\xspace}
\newcommand{\TTG}{TTG\xspace}
\newcommand{\TTGs}{TTGs\xspace}
\newcommand{\pc}{parallel composition\xspace}

\newcommand{\Pc}{Parallel composition\xspace}

\newcommand{\Sc}{Series composition\xspace}
\newcommand{\ttspg}{two-terminal series-parallel graph\xspace}
\newcommand{\Ttspg}{Two-terminal series-parallel graph\xspace}

\newcommand{\TTSPG}{TTSPG\xspace}

\newcommand{\Sns}{Simple node sequence\xspace}
\newcommand{\sns}{simple node sequence\xspace}
\newcommand{\snss}{simple node sequences\xspace}

\newcommand{\msc}{series composition\xspace}

\newcommand{\sed}{S-decomposition\xspace}

\newcommand{\apath}{A-path\xspace}
\newcommand{\apaths}{A-paths\xspace}
\newcommand{\spath}{S-path\xspace}
\newcommand{\spaths}{S-paths\xspace}

\newcommand{\spgapproxratio}{$14 \cdot D^2$\xspace}

\newcommand{\ARR}{\pi}
\newcommand{\RCOST}{R\text{-}COST}
\newcommand{\ACOST}{E\text{-}COST}
\newcommand{\VCOST}{A\text{-}COST}
\newcommand{\ARRALG}{\ARR_{ALG}}
\newcommand{\ARROPT}{\ARR_{OPT}}
\newcommand{\RARR}[1]{\ARR(#1)}
\newcommand{\RALG}[1]{\ARRALG(#1)}
\newcommand{\ROPT}[1]{\ARROPT(#1)}

\newcommand{\SED}{SD}

\newcommand{\mst}{-}
\newcommand{\mt}{\ominus}

\newcommand{\minLAP}{\textsc{minLA}\xspace}

\newcommand{\minlap}{minimum linear arrangement problem\xspace}

\newcommand{\spt}{SP-tree\xspace}
\newcommand{\mspt}{minimal SP-tree\xspace}

\newcommand{\lnode}{$L$\text{-}node\xspace}
\newcommand{\pnode}{$P$\text{-}node\xspace}
\newcommand{\snode}{$S$\text{-}node\xspace}
\newcommand{\lnodes}{$L$\text{-}nodes\xspace}
\newcommand{\pnodes}{$P$\text{-}nodes\xspace}
\newcommand{\snodes}{$S$\text{-}nodes\xspace}

\usepackage[top=4.20cm, bottom=4.25cm, left=4.75cm, right=4.05cm]{geometry}

\usepackage{url}

\begin{document}


\title{Minimum Linear Arrangement of Series-Parallel Graphs\thanks{This work was partially supported by the German Research Foundation (DFG) within the Collaborative Research Center ``On-The-Fly Computing'' (SFB 901).}}

\author{Martina Eikel \and Christian Scheideler\and Alexander Setzer}

\begin{titlepage}

\author{Martina Eikel \\
   University of Paderborn, Paderborn, Germany \\
   martinah@upb.de \\
   \and
   Christian Scheideler \\
   University of Paderborn, Paderborn, Germany\\
   scheideler@upb.de
   \and
   Alexander Setzer \\
   University of Paderborn, Paderborn, Germany \\
   asetzer@mail.upb.de \\
   }

\date{}
\end{titlepage}

\maketitle \thispagestyle{empty}


 \begin{abstract}
 We present a factor $14D^2$ approximation algorithm for the minimum linear arrangement problem on series-parallel graphs, where $D$ is the maximum degree in the graph.
 Given a suitable decomposition of the graph, our algorithm runs in time $O(|E|)$ and is very easy to implement.
 Its divide-and-conquer approach allows for an effective parallelization.
 Note that a suitable decomposition can also be computed in time $O(|E|\log{|E|})$ (or even $O(\log{|E|}\log^*{|E|})$ on an EREW PRAM using $O(|E|)$ processors).
 
 For the proof of the approximation ratio, we use a sophisticated charging method that uses techniques similar to amortized analysis in advanced data structures.

 On general graphs, the minimum linear arrangement problem is known to be NP-hard.
 To the best of our knowledge, the minimum linear arrangement problem on series-parallel graphs has not been studied before.

  \end{abstract}

  \newgeometry{top=2.5cm, bottom=2.5cm, left=2.70cm, right=3.15cm}

\section{Introduction}
The minimum linear arrangement problem is a well-known graph embedding problem, in which an arbitrary graph is mapped onto the line topology, such that the sum of the distances of nodes that share an edge is minimized.
We consider the class of series-parallel graphs, which arises naturally in the context of parallel programs: 
modelling the execution of a parallel program yields a series-parallel graph, where sources of parallel compositions represent fork points, and sinks of parallel compositions represent join points (for the definition of a parallel composition, see Subsection~\ref{spg:sec:definitions}).
Note that in this context, \spgs typically have a very low node degree:
Since spawning child processes is costly, one would usually not spawn too many of them at a time.

\subsection{Problem statement and definitions}\label{spg:sec:definitions}
Throughout this work, we consider undirected graphs only.
The following definition of the \minlap is based on \cite{bib:petit2003}:
\begin{definition}[Linear arrangement]
 Given a graph $G=(V,E)$, let $n = |V|$.
	A \emph{linear arrangement} $\ARR$ of $G$ is a one-to-one function
		\[ \ARR : V \rightarrow \{1, \dots, n\}. \]
	For a node $v \in V$, $\ARR(v)$ is also called the \emph{position of $v$} in $\ARR$.
\end{definition}
\begin{definition}[Cost of a linear arrangement]
 Given a graph $G=(V,E)$ and a linear arrangement $\ARR$ of $G$, we denote the \emph{cost} of $\ARR$ by
	\[ COST_{\ARR}(G) := \sum_{\{u,v\} \in E}{|\ARR(u) - \ARR(v)|}. \]
\end{definition}
\begin{definition}[Minimum linear arrangement problem]
 Given a graph $G=(V,E)$ (the \emph{input graph}), the \emph{\minlap} (\minLAP) is to find a linear arrangement $\ARR$ that minimizes $COST_{\ARR}(G)$.
\end{definition}

Next we define the class of \spgs, which we need several definitions for (the following is based on \cite{bib:eppstein1992}):
  \begin{description}
	 \item[\Ttg (\TTG)] A \emph{\ttg} $G=(V,E)$ is a graph with node set $V$, edge set $E$, and two distinct nodes $s_G, t_G \in V$ that are called source and sink, respectively.
		$s_G$ and $t_G$ are also called the \emph{terminals} of $G$.
	\item[\Sc]
	    The \emph{\msc} $SC$ of $k \geq 2$ \TTGs $X_1,\dots, X_k$ is a \TTG created from the disjoint union of $X_1,\dots,X_k$ with the following characteristics: 
	    The sink $t_{X_i}$ of $X_i$ is merged with the source $s_{X_{i+1}}$ of $X_{i+1}$ for $1 \leq i < k$.
	    The source $s_{X_1}$ of $X_1$ becomes the source $s_{SC}$ of $SC$ and the sink $t_{X_k}$ of $X_k$ becomes the sink $t_{SC}$ of $SC$.
	\item[\Pc]
	    The \emph{\pc} $PC$ of $k \geq 2$ \ttgs $X_1,\dots,X_k$ is a \TTG created from the disjoint union of $X_1,\dots,X_k$ with the following two characteristics:
	    The sources $s_{X_1},\dots,s_{X_k}$ are merged to create $s_{PC}$ and
	    the sinks $t_{X_1},\dots,t_{X_k}$ are merged to create $t_{PC}$.
	 \item[\Ttspg (\TTSPG)]\label{spg:def:ttspg}
	    A \emph{\ttspg} $G$ with source $s_G$ and sink $t_G$ is a graph that may be constructed by a sequence of series and parallel compositions starting from a set of copies of a single-edge \ttg $G'=(\{s,t\},\{\{s,t\}\})$.
	    \item[\Spg]
	    A graph $G$ is a \emph{\spg} if, for some two distinct nodes $s_G$ and $t_G$ in $G$, $G$ can be regarded as a \TTSPG with source $s_G$ and sink $t_G$.
 \end{description} 	    
Note that the series and parallel compositions are commonly defined over two input graphs only.
However, it is not hard to see that our definition of a \spg is equivalent.

An example of a \spg is shown in Figure~\ref{spg:fig:spg_example_big}.

\subsection{Related work}
	The \minLAP was first stated by Harper \cite{bib:harper1964}.
	Garey, Johnson, and Stockmeyer were the first to prove its NP-hardness on general graphs \cite{bib:garey1976}.
	Amb{\"u}hl, Mastrolilli, and Svensso showed that the \minLAP on general graphs does not have a polynomial-time approximation scheme unless NP-complete problems can be solved in randomized subexponential time \cite{bib:ambuhl2007}.	
	To the best of our knowledge, the two best polynomial-time approximation algorithms for the \minLAP on general graphs are due to Charikar, Hajiaghayi, Karloff, and Rao \cite{charikar2006}, and Feige and Lee \cite{bib:feige2007}.
	Both algorithms yield an $O(\sqrt{\log{n}}\log\log{n})$-approximation of the \minLAP.
	The latter algorithm is a combination of techniques of earlier works by Rao and Richa \cite{bib:rao1998}, and Arora, Rao, and Vazirani \cite{bib:arora2009}. 	
 	For planar graphs (which include the series-parallel graphs), Rao and Richa \cite{bib:rao1998} also present a $O(\log\log{n})$-approximation algorithm.
 	Note that even though, for high degree graphs, these algorithms achieve a better approximation factor than the one we present in this work, there are some key differences between these algorithms and ours: 
 	First of all, the algorithm we present is a very simple divide-and-conquer algorithm and its functioning can be understood easily. 
 	The aforementioned algorithms, however, are much more complex and involve solving a linear or semidefinite program.
 	Furthermore, our algorithm achieves a runtime of only $O(|E|)$ (if the \spg is given in a suitable format - otherwise, a more complex preprocessing is required that takes time $O(|E|\log{|E|})$, but this can be parallelized down to $O(
	\log{|E|}\log^*{|E|})$) making it suitable in situations where a low runtime is more important than the approximation guarantee.
 	Still, for low graph degrees (which are reasonable to assume in certain applications), our algorithm even improves the approximation factor of Rao and Richa.

 	For special classes of graphs, the NP-hardness has been shown for bipartite graphs \cite{bib:even1975}, interval graphs, and permutation graphs \cite{bib:cohen2006}.
	On the other hand, polynomial-time optimal algorithms have been found for hypercubes \cite{bib:harper1964}, trees \cite{bib:chung1988}, $d$-dimensional $c$-ary cliques \cite{bib:nakano1994}, meshes \cite{bib:fishburn2000}, and chord graphs \cite{bib:rostami2008}.
	Note that many people claim that the \minLAP is optimally solvable on outerplanar graphs, referring to \cite{bib:frederickson1988}.
	However, the problem solved in \cite{bib:frederickson1988} is different from the \minLAP as we show in \cite{bib:setzer2014}.
	Note that the question whether the \minLAP is NP-hard on \spgs is unsettled.
	  
 	Applications of the \minLAP include the design of error-correcting codes \cite{bib:harper1964}, machine job scheduling (e.g., \cite{bib:adolphson1977}), VLSI layout (e.g., \cite{bib:adolphson1973,bib:diaz2002}), and graph drawing (e.g., \cite{bib:shahrokhi2001}).
  	For an overview of heuristics for the \minLAP see the survey paper by Petit \cite{bib:petit2011}.	

 	The class of \spgs, first used by MacMahon \cite{bib:macmahon1892}, has been studied extensively.	  
	It turns out that many problems that are NP-complete on general graphs can be solved in linear time on \spgs.
	Among these are the decision version of the dominating set problem \cite{bib:kikuno1983}, the minimum vertex cover problem, the maximum outerplanar subgraph problem, and the maximum matching problem \cite{bib:takamizawa1982}.
	Furthermore, since the class of \spgs is a subclass of the class of planar graphs, any problem that is already in $P$ for that class of graphs can be solved optimally in polynomial time for \spgs as well (such as the max-cut problem \cite{bib:hadlock1975}).

      Another problem regarding \spgs is to decide, given an input graph $G$, whether it is series-parallel and, if so, to output the operations that recursively constructed the series-parallel graph.
      The first step is referred to as \emph{series-parallel graph recognition} while the second step is referred to as \emph{constructing a decomposition tree}.
      A parallel linear-time algorithm for this problem on directed graphs was first presented by Valdes, Tarjan, and Lawler \cite{bib:valdes1979}.
      Later, Eppstein \cite{bib:eppstein1992} developed a parallel algorithm for undirected graphs using a so-called \emph{nested ear decomposition}.
      The concept of an \sed used in our analysis is technically similar to that concept, though we use a different notation more suitable for our purposes.
      The algorithm we propose for approximating the \minLAP on \spgs also relies on a decomposition tree.
      For instances in which it is not given, the algorithm by Bodlaender and De Fluiter \cite{bib:bodlaender1996} can be used, since it runs on undirected graphs and outputs so-called \emph{\spt}, which can be easily transformed into a format suitable for our algorithm.

\subsection{Our contribution}
We describe a simple approximation algorithm for the \minlap on \spgs with an approximation ratio of $14D^2$, where $D$ is the degree of the graph, and a running time of $O(|E|)$ if the \spg is given in a suitable format. 
If the \spg is not given in the required format, this format can be computed in time $O(|E| \log{|E|})$ (which can even be further parallelized down to $O(\log{|E|}\log^*{|E|})$ on an EREW PRAM using $O(|E|)$ processors).
However, for certain applications it is reasonable to assume that the graph is given in the right format, e.g., when the series-parallel graph is used to model the execution of a parallel program, the desired representation can be constructed along with the model.
The simplicity and the structure of the algorithm allow for an efficient distributed implementation.

Moreover, our proof of the approximation ratio introduces a sophisticated charging method following an approach that is known from the amortized analysis of advanced data structures.
This technique may be applied in other analyses as well.

\section{Preliminaries}
The algorithm we present is defined recursively and is based on a decomposition of the \spg into components.
Therefore, prior to describing the algorithm, we introduce several definitions needed to formalize this decomposition.

The following definition is similar to the one in \cite{bib:bodlaender1996}.
\begin{definition}[\spt, \mspt]
 An \spt $T$ of a \spg $G$ is a rooted tree with the following properties:
  \begin{enumerate}
   \item Each node in $T$ corresponds to a two-terminal subgraph of $G$.
   \item Each leaf is a so-called \emph{\lnode} labelled as $L(k)$ and corresponds to a path with $k$ edges.
   \item Each inner node is a so-called \emph{\snode} or \emph{\pnode}, and the two-terminal subgraph $G'$ associated with an \snode (\pnode) is the graph obtained by a series (parallel) composition of the graphs associated with the children of $G'$, where the order of the children defines the order in which the series composition is applied (the order does not matter for a parallel composition).
   \item The root node corresponds to $G$.
  \end{enumerate}
 
 An \spt $T$ of a \spg $G$ is called \emph{minimal} if the following two conditions hold:
 \begin{enumerate}
  \item All children of an \snode are either \pnodes or \lnodes, but at least one is a \pnode.
  \item All children of a \pnode are either \snodes or \lnodes.
 \end{enumerate}  
\end{definition}
It is easy to see that for any fixed \spg $G$, there exists a \mspt for $G$.

We are now ready to introduce the following three important notions:
\begin{definition}[\Sns, Parallel component, Series component]
  Let $G$ be a \spg and $T$ be a \mspt of $G$.
   The sub-graph of $G$ associated with a leaf $L(k)$ of $T$ for $k \in \mathbb{N}$ is called a \emph{\sns}.
   The sub-graph of $G$ associated with a \pnode is called a \emph{parallel component} of $G$.
   The sub-graph of $G$ associated with a \snode is called a \emph{series component} of $G$.
   Furthermore, any \sns is called a \emph{series component}, too.
\end{definition}
An illustration of the different types of components is given by Figure~\ref{spg:fig:spg_example_big}.
\begin{figure}[tb]\centering
    \includegraphics[width=.3\linewidth]{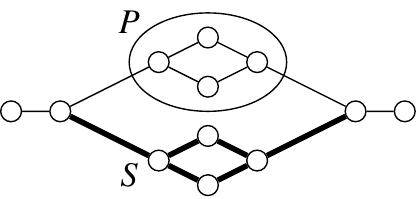}
    \caption{Example of a simple \spg.
	      $P$ is a parallel component consisting of two series components (more precisely, two \snss with two edges each).
	      The thick edges belong to the subgraph induced by the series component $S$.
	      It consists of two single-edge \snss (on the left and right end) and a parallel component.}
	\label{spg:fig:spg_example_big}
\end{figure}
The definition of a \mspt implies the following:
Each parallel component $P$ is the result of a parallel composition of two or more series components.
Furthermore, each series component $S$ is the result of a series composition of two or more parallel components or \snss, but not exclusively \snss.
This leads to the following definition:
\begin{definition}[Child component]
  Let $G$ be a \spg, let $T$ be a \mspt, and let $X$ and $Y$ be two nodes in $T$ such that $Y$ is a child of $X$.
  Further, let $C_i$ be the (series or parallel) component that is associated with $Y$ and let $C$ be the (parallel or series) component $C$ that is associated with $X$.
  Then, $C_i$ is called a \emph{child component} of $C$, and we say: $C_i \in C$.
\end{definition}
For example, the two \snss that induce the parallel component $P$ in Figure~\ref{spg:fig:spg_example_big} are child components of $P$. 
One implication of this definition is that the terminals of a parallel component and its child components overlap.

For the rest of this work, we assume that for any fixed \spg $G$, the \snss, series components and parallel components of $G$ are uniquely defined by a fixed \mspt $T$.
In Appendix~\ref{sec:appendix:scs}, we describe an efficient method to compute a \mspt according to our definition.
It is basically an extension of an algorithm by Bodlaender and de Fluiter \cite{bib:bodlaender1996}.

\section{The series-parallel graph arrangement algorithm}\label{spg:sec:algorithm}
The \spgalglong (\spgalg) is defined recursively.
In order to arrange the nodes of a series or parallel component $C$, the \Spgalg first determines the order of its child components recursively, and then places the child components side by side in an order that depends on their size.
For any (series or parallel) component $C$, when the algorithm has just arranged the nodes of $C$, it holds that its source receives the leftmost position among all nodes of $C$ and that its sink receives the position directly to the right of the source.
However, later computations (in a higher recursion level) may re-arrange the terminals and pull them apart.
More specific details are given in the corresponding subsections for the different types of components.

Illustrations of all arrangements and all different cases can be found in Appendix~\ref{sec:appendix:figures}.

\subsection{Arrangement of a \sns}
\label{spg:subsec:alg:sns}
For any \sns $L$, we label the nodes of $L$ from left to right by $1$ to $k$.
That is, the source receives label $1$ and the sink receives label $k$.
The arrangement of this sequence then is: $1, k, 2, k-1, 3, k-2, \dots$.
One can see that this arrangement fulfills the property that the source is on the leftmost position and that the sink is its right neighbor.

\subsection{Arrangement of a parallel component}
\label{spg:subsec:alg:pc}
For any parallel component $P$ with source $u$, sink $v$, and $m \geq 2$ child components $S_1, S_2, \dots, S_m$ (note that any parallel component has at least two child components), the \Spgalg recursively determines the arrangement of the child components.
We denote the computed arrangement of $S_i$ excluding the two terminal nodes (which would have been placed at the first two positions of the arrangement, see Subsection~\ref{spg:subsec:alg:sc}) by $S_i^\mst$.
\Wlog let $S_m$ be a biggest child component (\wrt the number of nodes in it).
Then, the algorithm places $u$ at the first position, $v$ at the second position, the nodes of $S_1^\mst$ to the right of that (in their order), and the nodes of $S_i^\mst$ to the right of $S_{i-1}^\mst$ for $i \in \{2,\dots,m\}$.

\subsection{Arrangement of a series component}
\label{spg:subsec:alg:sc}
For any series component $S$ with source $u$, sink $v$, and $m \geq 2$ child components $P_1, P_2, \dots, P_m$ (note that any series component has at least two child components, otherwise it would be a \sns), the \spgalg first recursively determines the arrangement of the child components.
Second, it puts $u$ and $v$ at the first two positions, in this order.
The third step differs from the case of a parallel component:
To keep the cost of the arrangement low while ensuring that a biggest child component $P_a$ receives the rightmost position, the general order of the child components is: $P_1, P_2, \dots, P_{a-1}, P_m, P_{m-1}, \dots, P_{a+2}, P_{a+1}, P_a$.
Here, the components from $P_m$ to $P_{a+1}$ are flipped (the order of their nodes is reversed).
For $m=a$, the order is $P_1, P_2, \dots, P_m$ and for $a=1$, the components are ordered in reverse (i.e., $P_m, P_{m-1}, \dots, P_1$) (where all components except for $P_1$ are flipped).

However, since each two neighboring child components $P_i$ and $P_{i+1}$ share a (terminal) node, it must be decided which of the two components may ``keep'' its node.
The strategy here is as follows: Each component $P_i$ (except for the first component, whose source has received the leftmost position already) keeps its source and lends its sink to $P_{i+1}$ (of which it is a source), except for $P_m$ (whose sink has been placed at the second position already).
This may stretch existing edges, which we will keep track of in the analysis.

An illustration of the arrangement for the case $1 < a < m$ can be found in Figure~\ref{spg:fig:alg:sc}.

\begin{figure}[htb]\centering
		\includegraphics[width=.9\linewidth]{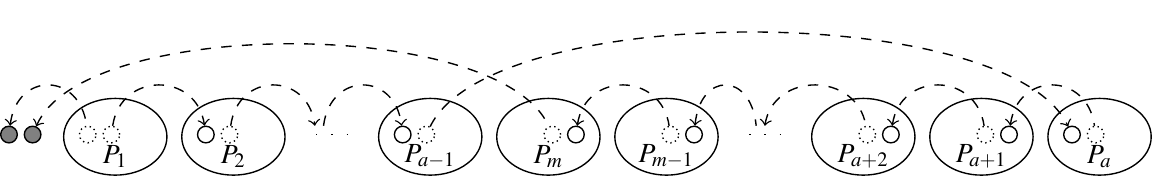}
	\caption{Order in which the \Spgalg arranges a series component consisting of $m$ child components for $1 < a < m$ (where $P_a$ is a biggest component).
	Dotted nodes indicate the position at which a node would be placed according to the previous recursion level.
	Dashed arrows indicate the change in position at the current recursion level.}
 \label{spg:fig:alg:sc}
\end{figure}

\section{Analysis}\label{sec:analysis}
In this section, we prove the approximation ratio of \spgapproxratio for the \spgalglong described in Section~\ref{spg:sec:algorithm}.
  As a first step, we provide lower bounds on the \emph{\vaddcost} in an optimal arrangement for each kind of component.
The \vaddcost of a component is the sum of two values:
First, the \addcost of this component (cost of the current component minus the individual cost of all child components).
Second, some cost that has been accounted for in a lower recursion level.
This cost is chosen such that the sum of all \vaddcosts  does not contain this cost more than three times. 
We use these bounds to establish a lower bound on the total cost of an optimal solution.
The details are described in Subsection~\ref{spg:subsec:ana:lowerbound}.
As a second step, we state upper bounds on the \addcosts generated at each recursion step of the \spgalg in order to determine an upper bound on the total cost in Subsection~\ref{spg:subsec:ana:upperbound}.
Last, we use both the lower bound as well as the upper bound to relate the cost of an optimal arrangement to that of an arrangement computed by the \spgalg.
This is done in Subsection~\ref{spg:subsec:ana:final}.
In addition to providing the approximation ratio of the \spgalg, we establish a polynomial runtime bound of our algorithm in Subsection~\ref{spg:subsec:ana:runtime}.
Note that all the proofs in this section can be found in Appendix~\ref{sec:appendix:proof}.

\subsection{Prerequisites}
 \label{spg:subsec:prereq}
 For the analysis, we need several notions, which we now introduce.

\begin{definition}[Length of an edge] 
	Given a graph $G=(V,E)$ and a linear arrangement $\ARR$ of $G$, let $u,v \in E$.
	The \emph{length of $(u,v)$ in $\ARR$}, denoted by $length_{\ARR}(u,v)$ is defined as:
  \[ length_{\ARR}(u,v) = |\ARR(u) - \ARR(v)|. \]
\end{definition}

\begin{definition}
 Given a linear arrangement $\ARR$ of a \spg $G=(V,E)$ and a (series or parallel) component $C$ in $G$, we define:
	\begin{description}
	 \item[Restricted arrangement.] The \emph{arrangement $\ARR$ restricted to $C$}, denoted by $\RARR{C}$ is obtained by removing all nodes from $\ARR$ that do not belong to $C$, as well as their incident edges, i.e., $\RARR{C}$ maps the nodes from $C$ to $\{1, \dots, |C|\}$.
	 \item[Restricted length of an edge.] For any edge $(u,v)$ that belongs to $C$, the \emph{length of $(u,v)$ restricted to $C$}, denoted by \emph{$length_{\RARR{C}}(u,v)$}, is the distance between $u$ and $v$ in $\RARR{C}$.
	 \item[Restricted cost of an arrangement.] Let $E_C$ be the set of all edges from $G$ whose both endpoints are in $C$.
		The \emph{cost of $C$ restricted to $C$}, denoted by \emph{$\RCOST_\ARR(C)$}, is defined as: \[\RCOST_\ARR(C) := \sum_{(u,v) \in E_C}{length_{\RARR{C}}(u,v)}.\]
	\end{description}
\end{definition}

\begin{definition}[\Addcost of a series / parallel component]
\label{spg:def:addcostcomponent}
 Given a linear arrangement $\ARR$ of a \spg $G$ and a (series or parallel) component $C$ in $G$ containing $m \geq 0$ child components $C_1, \dots, C_m$,
	the \emph{\addcost of $C$ in $\ARR$}, denoted by $\ACOST_\ARR(C)$, is defined as
		\[ \ACOST_\ARR(C) := \RCOST_\ARR(C) - \sum_{i=1}^{m}{\RCOST_\ARR(C_i)}. \]
\end{definition}
Note that the \addcost of a \sns $S$ is equal to the restricted cost of $S$.

We can make the following observation regarding the relationship between the \addcosts of the components and the total cost:
\begin{observation}\label{spg:obs:totalcost}
	Let $G$ be a \spg and let $\ARR$ be a linear arrangement of $G$.
	Further, let $\mathcal{C}$ be the set of all (series or parallel) components in $G$.
	It holds:
			\[ \sum_{C \in \mathcal{C}}{\ACOST_\ARR(C)} = COST_\ARR(G). \]  
\end{observation}

In the analysis of the \spgalg, we need to find at least one path from $s_P$ to $t_P$ through $P$ for each parallel component $P$ such that any two such paths are edge-disjoint for two different parallel components.
Therefore, we introduce the following notion of an \emph{\sed}, which yields these paths and is recursively defined as follows:
\begin{definition}[\apath, \spath, \sed
]\label{spg:def:sed}
	Let $P$ be an ``innermost'' parallel component in a \spg $G$ (i.e., one whose child components are \snss only) with source $s$, sink $t$, and $k$ child components.
	Select an arbitrary simple path from $s$ to $t$ through $P$ (i.e., select one of the \snss).
	This path is called the \emph{auxiliary path} or simply \emph{\apath} of $P$.
	The remaining paths from $s$ to $t$ through $P$ are called the \emph{selected paths} or simply \emph{\spaths} of $P$.
	
	Recursively, for an arbitrary parallel component $P$, with source $s$, sink $t$, and $m \geq 2$ child components $S_1, \dots, S_m$, for each child component $S_i$, $1 \leq i \leq m$, select a simple path $Q_i$ from $s$ to $t$ through $S_i$ in the following way:
	If $S_i$ is a \sns, $Q_i$ is the whole sequence.
	Otherwise, $S_i$ is a series component, which consists of $k \geq 0$ simple node sequences and $l \geq 1$ parallel components (note that $k+l \geq 2$).
	Denote these child components by $P_1, \dots P_{k+l}$ in the order in which they appear in $S_i$.
	Construct the path $Q_i$ step by step:
	Start with $P_1$ and add $P_1$ completely to $Q_i$ if $P_1$ is a \sns.
	If, however, $P_1$ is a parallel component, select the \apath of $P_1$ and extend $Q_i$ by it.
	Continue in the same manner up to $P_{k+l}$.
	After this, the whole path $Q_i$ is constructed.
	$Q_1$ is called the \emph{\apath} of $P$ and the remaining paths $Q_2, \dots, Q_m$ are called the \emph{\spaths} of $P$.
 
	The selection of \spaths (and \apaths accordingly) for all parallel components of $G$ is called an \emph{\sed} of $G$.

\end{definition}
An example of an \sed can be found in Figure~\ref{spg:fig:sed_example} in Appendix~\ref{sec:appendix:figures}.

Intuitively, an auxiliary path of a parallel component $P_j$ is a path through the whole component which is \emph{reserved} to be used in higher recursion levels (to eventually become part of an \spath there).
Any edges of an \spath are not used for any \spath or \apath in any higher recursion level.

The main contribution of the \sed is that it gives a mapping from parallel components to paths through the respective components (the \spaths) such that all these paths are edge-disjoint.
More formally:

\begin{restatable}{lemma}{sedexistence}
\label{spg:lem:sed_existence}
	For each \spg $G$, there exists an \sed $\SED$.
	Besides, in any \sed, each edge belongs to at most one \spath in $\SED$.
\end{restatable}

Provided with the definition of an \sed, we are ready to define the \vaddcost as follows:
\begin{definition}[\Vaddcost]\label{spg:def:vaddcost}
	Let $\ARROPT$ be an (optimal) linear arrangement of a \spg $G$, let $\SED$ be an \sed of $G$, and let $S$ be a series component in $G$.
  Further, let $E_S$ be the set that contains all edges of \snss that are child components of $S$ and all edges of \spaths of the child components of $S$ that are parallel components.
	The \emph{\vaddcost of $S$}, denoted by $\VCOST_{\ARROPT}(S)$, is defined as:
	 \begin{equation*}
	      \VCOST_{\ARROPT}(S) := \ACOST_{\ARROPT}(S) + \sum_{\{x,y\}\in E_S}{length_{\ARROPT(S)}(x,y)}.
	 \end{equation*}
	
For any parallel component $P$ in a \spg $G$ and any optimal linear arrangement $\ARROPT$, 
	\[ \VCOST_{\ARROPT}(P) := \ACOST_{\ARROPT}(P).\]
\end{definition}
Note that the addend in the \vaddcost for \snss is zero (as the set $E_S$ is empty in this case).
As an example, the corresponding set $E_{S_i}$ for the series component $S_i$ from Figure~\ref{spg:fig:sed_example} would contain all normally drawn edges, the two dashed \spaths of $P_A$ and the lower path of $P_B$ as shown in Figure~\ref{spg:fig:sed_example}\subref{spg:fig:sed_examplec}.
Note that $E_{S_i}$ contains a path from $s_{S_i}$ to $t_{S_i}$.

This definition will be helpful for the analysis of the minimum cost of an optimal arrangement.
The \vaddcost adds a certain value to the \addcost of a (series or parallel) component $C$, with the following property:
\begin{restatable}{lemma}{sumofvaddcostatmostthreetimessumofaddcost}\label{spg:lem:sumofvaddcostatmostthreetimessumofaddcost}
	Let $G=(V,E)$ be a \spg, and $\ARROPT$ be an (optimal) linear arrangement for $G$.
	Further, let $\mathcal{C}$ be the set of all (series or parallel) components of $G$.
		It holds:
			\[ \sum_{C \in \mathcal{C}}{\VCOST_{\ARROPT}(C)} \leq 3\cdot \sum_{C\in\mathcal{C}}{\ACOST_{\ARROPT}(C)}. \]
\end{restatable}

For the analysis of an optimal arrangement, we also need the following notation:
\begin{definition}[$\Delta_C$]\label{spg:def:delta}
 Given an (optimal) linear arrangement $\ARROPT$ of a \spg $G$, and a (series or parallel) component $C$ in $G$, consider $\ARROPT$ restricted to $C$.
	We denote the smallest number of nodes to the left or to the right (depending on which number is smaller) of a terminal node of $C$ in $\ROPT{C}$ by $\Delta_C$.
\end{definition}

It is convenient to define:
\begin{definition}[Cardinality of a component]
  For any \spg $G$ and any (series or parallel) component $C$ in $G$: $|C|$ is the number of nodes in $C$, $|C^\mt|$ is the number of all nodes in $C$ without the sink of $C$, and $|C^\mst|$ is the number of nodes in $C$ without the two terminal nodes of $C$.
\end{definition}

The following lemma is easy to show:
\begin{lemma}\label{spg:lem:deltabound}
 Let $\ARROPT$ be an (optimal) linear arrangement of a \spg $G$ and let $C$ be a (series or parallel) component.
	It holds:
		\[ \Delta_C \leq \left\lfloor\frac{1}{2}(|C|-2)\right\rfloor. \]
\end{lemma}
%
%

Last, we need the following definition:
\begin{definition}[Spanning interval]\label{spg:def:spanninginterval}
	Given a graph $G$, an arrangement $\ARR$ of $G$, and a set $S$ of nodes, let $u \in S$ such that $\ARR(u) \leq \ARR(s)$ for all nodes $s \in S$, and let $v \in S$ such that $\ARR(v) \geq \ARR(s)$ for all $s \in S$.
	Let $I$ be the set of nodes such that $x \in I$ iff $\ARR(x) \in [\ARR(u), \ARR(v)]$.
	Then, $I$ is the \emph{interval spanning $S$ in $\ARR$}.
\end{definition}

\subsection{A lower bound on the total cost of optimal solutions}\label{spg:subsec:ana:lowerbound}
In this subsection, we give lower bounds on the \vaddcosts of an optimal arrangement for \snss, parallel components, and series components.
In the end, we consolidate the results and state a general lower bound on the total cost of an optimal arrangement.

\begin{restatable}{lemma}{optsns}\label{spg:lem:optsns}
 For any \sns $L$ in a \spg $G$ in an optimal arrangement $\ARROPT$, it holds:
	\[ \VCOST_{\ARROPT}(L) \geq |L| - 1 + \Delta_L. \]  
\end{restatable}

We now provide a lower bound on the \vaddcost of series components:
\begin{restatable}{lemma}{optsc}\label{spg:lem:optsc}
 For any series component $S$ in a \spg $G$ with $m \geq 2$ child components $P_1, \dots, P_m$ in an optimal arrangement $\ARROPT$, it holds:
	\[
	  \VCOST_{\ARROPT}(S) \geq \frac{1}{2}\left(\sum_{i=1}^m{|P_i^\mt|} - \max_i{|P_i^\mt|}\right) + 1 + \sum_{i=1}^m{\Delta_{P_i}} - \Delta_S.
	\]    
\end{restatable}

In the following, we present the roadmap for the proof of Lemma~\ref{spg:lem:optsns}. 
The full proof is given in Appendix~\ref{sec:appendix:proof}.

\begin{proofsketch}
 The idea of the proof is the following:
 First of all, we note that (by Definition~\ref{spg:def:vaddcost}) the \vaddcost of $S$ is the sum of two values: 
 the \addcost of $S$, and the sum over certain restricted edge lengths (the latter sum we denote by $\sigma$).
 Thus, we take a closer look at how the \addcost arises and at which edge lengths contribute to the value of $\sigma$.
       
    For the additional term $\sigma$, observe that it is the sum of lengths (restricted to $S$) of all edges in the set $E_S$. This set contains all edges of \snss that are child components of $S$ and all edges of \spaths of the parallel child components of $S$.
    
  For the \addcost of $S$, observe that they arise in the following way:
    An edge $e$ has a greater length in the arrangement $\ARROPT$ restricted to $S$ than in the arrangement $\ARROPT$ restricted to the child component $P_i$ of $S$ that $e$ is from.
    We call the difference in the length of $e$ the \emph{additional stretching} of $e$.
    Then, the sum of the additional stretchings of all edges is the \addcost of $S$.

    These two insights now enable use to determine a lower bound on the \vaddcost of $S$ by giving a lower bound on $\sigma$ and on the sum of additional stretchings individually.
    
    For the lower bound on $\sigma$, let $V_S$ be the set of nodes that contains all the endpoints of edges in $E_S$ and denote the interval spanning the set $V_S$ in $\ROPT{S}$ by $I$ (see Def.~\ref{spg:def:spanninginterval}).
    Since there exists a simple path from the leftmost node in $V_S$ (\wrt $\ROPT{S}$) to the rightmost node in $V_S$ consisting of edges from $E_S$ only, $\sigma$ is at least $|I|-1$.
    
    For the lower bound on the sum of the additional stretchings, we consider nodes whose positions are to the left of those in $I$ (the \emph{part left of $I$}) and nodes whose positions are to the right of those in $I$ (the \emph{part right of $I$}) individually.
    We show that for any node in the part left of $I$, except for nodes in a sequence of nodes that starts with the leftmost node (this sequence will be called $A$), at least one edge is stretched by an amount of one due to this node.
    The analog can be shown for the part right of $I$ (there, the corresponding sequence is called $B$).
    This yields a lower bound on the sum of additional stretchings linear in the number of these nodes (which is $|S|-|I|-|A|-|B|$). 
    
    All in all, this way we can find a lower bound on the amortized cost of $S$ that is $|S| - 1 - |A| - |B|$.
    Applying some upper bounds on the sizes of $A$ and $B$, we can rewrite this inequation as in the original claim of the lemma.
\end{proofsketch}

Using a similar technique as in the previous proof, we can also derive a lower bound for the \vaddcost of parallel components:
\begin{restatable}{lemma}{optpc}\label{spg:lem:optpc}
 For any parallel component $P$ in a \spg $G$ with $m \geq 2$ child components $S_1, \dots, S_m$, in an optimal arrangement $\ARROPT$, it holds:
	\[ \VCOST_{\ARROPT}(P) \geq \frac{1}{2}\left(\sum_{i=1}^m{|S_i^\mst|} - \max_i{|S_i^\mst|}\right) + \sum_{i=1}^m{\Delta_{S_i}} - \Delta_P. \]
\end{restatable}

These three lower bounds for the different types of components in any \spg can be combined into a single lower bound:
\begin{restatable}{corollary}{opttotal}\label{spg:cor:opttotal}
 Let $G=(V,E)$ be an arbitrary \spg and $\ARROPT$ an optimal arrangement of $G$.
	Further, denote the total cost of $\ARROPT$ by $COST_{\ARROPT}(G)$, the set of \snss in $G$ by $L_G$, the set of parallel components by $P_G$, the set of series components by $S_G$.
	Then, it holds:
	\begin{align*}
	 7 \cdot COST_{\ARROPT}(G) \geq& \sum_{L \in L_G}{2\cdot(|L|-1)}
														+ \sum_{P \in P_G}{\left(\sum_{S_i \in P}{|S_i^\mst|} - \max_{S_i \in P}{|S_i^\mst|}\right)}\\
														&+ \sum_{S \in S_G}{\left(\sum_{P_i \in S}{|P_i^\mt|} - \max_{P_i \in S}{|P_i^\mt|}\right)}. 
	\end{align*}
\end{restatable}

\subsection{An upper bound on the total cost of \spgalg arrangements}\label{spg:subsec:ana:upperbound}
For the approximation ratio of the \spgalg, we also need to find an upper bound on the cost of arrangements computed by the \spgalg.
One can show the following result:

\begin{restatable}{corollary}{algtotal}\label{spg:cor:algtotal}
 Let $G=(V,E)$ be an arbitrary \spg and let $\ARRALG$ be an arrangement of $G$ computed by the \spgalg.
	Furthermore, denote the total cost of $\ARRALG$ by $COST_{\ARRALG}(G)$, the set of \snss in $G$ by $L_G$, the set of parallel components by $P_G$, the set of series components by $S_G$.
	Then, it holds:
	\begin{align*}
		COST_{\ARRALG}(G) \leq& \sum_{L \in L_G}{2\cdot(|L|-1)}
			+ \sum_{P \in P_G}{2D^2\cdot\left(\sum_{S_i \in P}{|S_i^\mst|} - \max_{S_i \in P}{|S_i^\mst|}\right)}\\
		 &+ \sum_{S \in S_G}{2D\cdot\left(\sum_{P_i \in S}{|P_i^\mt|} - \max_{P_i \in S}{|P_i^\mt|}\right)}.
	\end{align*}
\end{restatable}

\subsection{The approximation ratio of \spgapproxratio}\label{spg:subsec:ana:final}
Finally, based on the groundwork of the previous subsections, proving the main theorem of this chapter is straightforward.
\begin{restatable}{theorem}{thmapproxratio}
	For a \spg $G$, let $\ARRALG$ be the linear arrangement of $G$ computed by the \spgalg, and let $\ARROPT$ be an optimal linear arrangement of $G$.
	It holds:
		\[ COST_{\ARRALG}(G) \leq 14 \cdot D^2 \cdot COST_{\ARROPT}(G). \]
\end{restatable}

\subsection{Runtime}\label{spg:subsec:ana:runtime}
Regarding the runtime of the \spgalg, one can show the following result:
\begin{restatable}{theorem}{runtime}\label{spg:lem:runtime}
 On a \spg $G=(V,E)$, the \spgalg has a runtime of $O(|E|)$ if a \mspt of $G$ is given as an input, and a runtime of $O(|E|\log{|E|})$ otherwise.
\end{restatable}

\bibliographystyle{plain}
\bibliography{literature}

\pagebreak

  \begin{appendix}

\section{Computing a minimal SP-Tree}\label{sec:appendix:scs}
In order to compute a \mspt for a given \spg $G$, the algorithm by Bodlaender and de Fluiter \cite{bib:bodlaender1996} can be used.
It is a parallel algorithm that has a running time of $O(\log{|E|})$ on a concurrent random-access machine (CRCW PRAM) using $O(|E|)$ processors for an input graph $G=(V,E)$.
Thus it can be simulated in time $O(|E|\log{|E|})$ by a sequential algorithm.
Furthermore, it has a runtime of $O(\log{|E|}\log^*{|E|})$ on an exclusive read exclusive write parallel random-access machine (EREW PRAM) using $|E|$ processors.

However, their definition (which we call Def.~A) of an \spt slightly differs from ours (which we call Def.~B):
According to Def.~A, each leaf in an \spt only represents a single edge, whereas in Def.~B, leaves may represent line-graphs with $k\geq 1$ edges.
However, it is not difficult to transform an \spt according to Def.~A to an \spt according to Def.~B: 
Simply transform each subtree induced by an \snode at the second lowest level and its children to an \lnode $L(k)$, with $k$ being the number of children of the \snode in the original \spt, and all remaining leaves that are neither \pnodes nor \snodes to \lnodes $L(1)$.

Note that since each parallel or series composition has at least two input graphs, any non-leaf in a \spt has at least two children.
Since each leaf in a \spt according to Def.~A corresponds to an edge in $G$, this implies that the size of the \mspt returned by the algorithm of Bodlaender and de Fluiter is at most $2|E|$ and that its height is at most $O(\log{|E|})$.
For the sequential case, the \spt transformation can obviously be performed in time linear in the size of the \spt that it is performed on.
Thus, determining a \mspt according to Def.~B is possible in time $O(|E|\log{|E|})$ in this case.
For the parallel case (execution on an EREW PRAM with $\Theta(|E|)$ processors), it is not difficult to see that the \spt transformation can be implemented such that the runtime is linear in the height of the \spt, thus determining a \mspt according to Def.~B is possible in time $O(\log{|E|}\log^*{|E|})$ in the parallel case.

\section{Additional proofs}\label{sec:appendix:proof}
\sedexistence*
\begin{proof}
Let $G$ be a \spg.
In order to show that the \sed is well-defined, i.e., we can find \apaths and \spaths as in Definition~\ref{spg:def:sed}, we need to show that there is an \apath for any parallel component (an assumption we make when constructing the \apath through a non-``innermost'' parallel component).
We can do this via induction on the number of nodes of a parallel component.
The induction hypothesis is: For any parallel component $P$, there is an \apath from the source of $P$ to the sink of $P$.

By Definition~\ref{spg:def:sed}, the induction hypothesis holds for all ``innermost'' components.
Let $P$ be a non-``innermost'' parallel component and let the induction hypothesis hold for all parallel components of smaller size.
In particular, the induction hypothesis holds for all (parallel) child components of (series) child components of $P$.
Thus, we can construct the paths $Q_i$, $1 \leq i \leq m$ through all $m$ child components of $P$.
Since we choose $Q_1$ to be the \apath of $P$, the claim also holds for $P$ and the induction is finished.

For the second claim, observe that by Definition~\ref{spg:def:sed}, whenever we select a path $Q$ to become an \apath or an \spath of a parallel component $P$, any section of this path has either been of no type before (which holds for the sections of the series child component $S$ of $P$ that are \snss) or has been an \apath of a next smaller parallel component.
This inductively implies that each edge can belong to at most one \spath in $\SED$.
\end{proof}

\sumofvaddcostatmostthreetimessumofaddcost*
\begin{proof}
	Note that the only difference between the definition of the \addcost and the \vaddcost is the addend ($\sum_{\{x,y\}\in E_C}{length_S(x,y)}$) in the definition of the \vaddcost of series components.
	Denote by $\sigma$ the sum of all these addends for all series components in $\mathcal{C}$.
	Then, by Definition~\ref{spg:def:vaddcost}: $\sum_{C \in \mathcal{C}}{\VCOST_{\ARROPT}(C)} = \sum_{C\in\mathcal{C}}{\ACOST_{\ARROPT}(C)} + \sigma$.
	Thus the claim of the lemma is equivalent to $\sigma \leq 2\cdot \sum_{C\in\mathcal{C}}{\ACOST_{\ARROPT}(C)}$.

	By Observation~\ref{spg:obs:totalcost}, this boils down to that $\sigma$ is at most twice the total optimal cost.
	We prove the claim by showing that the length of each edge $e \in E$ is taken into account at most two times in $\sigma$.
	Since the length of each edge is contained exactly once in the total optimal cost (by the definition of the total cost), this yields the claim.
	More specifically, we show that there are at most two distinct occasions, at which the length of an edge $e$ is considered for a calculation of the second addend.

	Consider a fixed edge $e \in E$.
	We now identify the only possible cases where the length of $e$ \wrt $\ARROPT$ (possibly restricted to some component) is added to $\sigma$.
	The first of the two cases in which this is possible is that $e$ is contained in a \sns that is a child component of a series component $S$.
	This can occur only once.
	The second case is that $e$ is contained in an \spath of a parallel component $P$ and taken into account for the calculation of the \vaddcost of $S$, the series component whose child is $P$.
	As follows from Lemma~\ref{spg:lem:sed_existence}, each edge can occur in one \spath only.
	Thus, also this case can occur only once.
\end{proof}

\optsns*
\begin{proof}
	Let $L$ be a \sns in a \spg $G$ and let $\ARROPT$ be an optimal linear arrangement of $G$.
	In the following, we consider the arrangement $\ARROPT$ restricted to $L$.
	\Wlog let $u$ be the node that defines $\Delta_L$ in $\ROPT{L}$, i.e., the terminal node with a minimum number of nodes to its right or left in $\ROPT{L}$.
	Furthermore, \Wlogs assume that $u$ is on the left half of $L$ in $\ROPT{L}$.	
	Denote the leftmost node in $\ROPT{L}$ by $l$ and the rightmost node in $\ROPT{L}$ by $r$.
	An illustration of the given situation can be found in Figure~\ref{spg:fig:lb_sns}\subref{spg:fig:lb_sns_a}.
	There is a simple path $Q$ from $u$ to $l$ in $L$.
	This path may or may not contain $r$.

	\begin{figure}[htb]\centering
	      \subfloat[width=170pt][Schematic representation of $\ROPT{L}$ (edges are not shown).]{
	      \label{spg:fig:lb_sns_a}
	      \centering		      
		      \includegraphics[width=170pt]{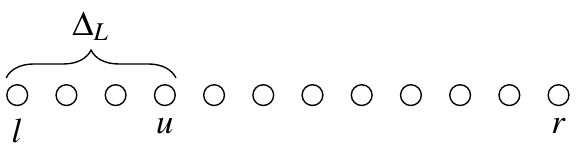}
	      }
	      
	      \subfloat[width=170pt][Example of the original \sns $L$ in $G$ as in the first case.]{
		  \label{spg:fig:lb_sns_b}
		    \centering		  
		      \includegraphics[width=170pt]{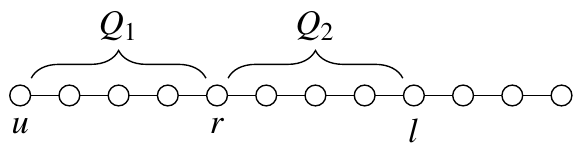}
	      }
	      \hfill
	      \subfloat[width=170pt][Example of the original \sns $L$ in $G$ as in the second case.]{
		    \label{spg:fig:lb_sns_c}
		    \centering
		      \includegraphics[width=170pt]{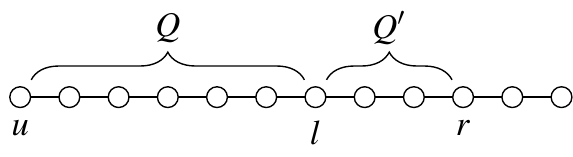}
	      }
	      \caption{Illustration of the notions of Lemma~\ref{spg:lem:optsns}.}
	      \label{spg:fig:lb_sns}
	\end{figure}

	We first consider the case that $Q$ contains $r$.
	In this case, split $Q$ into two paths $Q_1$ and $Q_2$, the first ranging from $u$ to $r$, the second ranging from $r$ to $l$.
	This is illustrated in Figure~\ref{spg:fig:lb_sns}\subref{spg:fig:lb_sns_b}.
	The first path has a length of at least $|L|/2 \geq \Delta_L$ in $\ROPT{L}$, whereas the second path has a length of at least $|L|-1$ in $\ROPT{L}$.
	Since any simple node sequence consists of a single simple path only, no edge is counted twice in this calculation.
	All in all, this completes the proof for this case.

	If $Q$ does not contain $r$, then there is also a second path $Q'$ from $l$ to $r$ whose nodes are disjoint from those in $Q$ (except for $l$).
 	This is illustrated in Figure~\ref{spg:fig:lb_sns}\subref{spg:fig:lb_sns_c}.
	The length of $Q'$ in $\ROPT{L}$ is at least $|L|-1$, whereas the length of $Q$ in $\ROPT{L}$ is at least $\Delta_L$.
	Together, this yields the claim in this case, too.
\end{proof}

\optsc*
\begin{proof}
  Let $S$ be a series component with $m \geq 2$ child components $P_1, \dots, P_m$ in a \spg $G$ and let $\ARROPT$ be an optimal linear arrangement of $G$.
	In the following, we consider the arrangement $\ARROPT$ restricted to $S$.
	
	The \vaddcost of $S$ is defined as the sum of two values:
	the \addcost of $S$ and a sum over certain edge lengths (the latter sum we denote by $\sigma$).
	We now take a closer look at how the \addcost arise and which edge lengths contribute to the value of $\sigma$.
			
	The additional term $\sigma$ is the sum of lengths (restricted to $S$) of all edges in the set $E_S$ that contains all edges of \snss that are child components of $S$ and all edges of \spaths of the child components of $S$ that are parallel components.
	
	The \addcost of $S$ arise in the following way:
	An edge $e$ has a greater length in the arrangement $\ARROPT$ restricted to $S$ than in the arrangement $\ARROPT$ restricted to the child component $P_i$ of $S$ that $e$ is from.
	We call the difference in the length of $e$ the \emph{additional stretching} of $e$.
	Then, the sum of additional stretchings of all edges is the \addcost of $S$.
	Thus, for the \addcost we need to determine a lower bound on the sum of additional stretchings due to $S$.		
	
	We now give a lower bound on $\sigma$ and on the sum of additional stretchings individually.
	The sum of these two bounds yields a lower bound of the \vaddcost of $S$.
	
	Let $V_S$ be the set of nodes that contains all the endpoints of edges in $E_S$ and denote the interval spanning the set $V_S$ in $\ROPT{S}$ by $I$ (see Def.~\ref{spg:def:spanninginterval}).
	Since there exists a simple path from the leftmost node in $V_S$ (\wrt $\ROPT{S}$) to the rightmost node in $V_S$ consisting of edges from $E_S$ only, $\sigma$ is at least $|I|-1$.

	In order to prove the lower bound on the sum of the additional stretchings, we consider nodes whose positions are to the left of those in $I$ (the \emph{part left of $I$}) and nodes whose positions are to the right of those in $I$ (the \emph{part right of $I$}) individually.
	We show that for any node in the part left of $I$, except for nodes in a sequence of nodes that starts with the leftmost node, at least one edge is stretched by an amount of one due to this node.
	The analog can be shown for the part right of $I$.
	This yields a lower bound on the sum of additional stretchings linear in the number of these nodes.

	Denote the leftmost node of $\ROPT{S}$ by $a$.
	In the following, assume $a \notin V_S$.
	Later on, we will consider the case $a \in V_S$ as well.
	Let $A$ be the maximal sequence of consecutive (\wrt $\ROPT{S}$) nodes that starts with $a$ and contains only nodes that belong to the same child component $P_A$ of $S$ as $a$, but not to $V_S$.
	Note that $P_A$ may be either a \sns or a parallel component.
	Similarly, assume that the rightmost node of $\ROPT{S}$ is not in $V_S$ (again, we will consider the other case later on).
	Let $B$ be the analogous sequence starting with the rightmost node and let $P_B$ be defined analogously (i.e., $B$ must not include any node from $V_S$).
	The situation is depicted in Figure~\ref{spg:fig:ana:lowerbound:sc_situation}.
	\begin{figure}[htb]\centering
		\includegraphics[width=.6\linewidth]{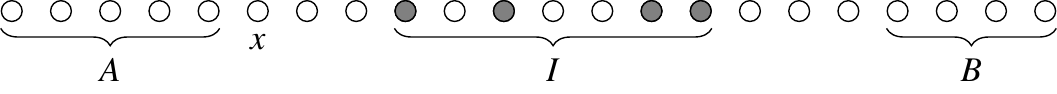}
		\caption{Example of $\ARROPT$ restricted to $S$ in the proof of Lemma~\ref{spg:lem:optsc}. 			
						Gray nodes are from $V_S$.}
		\label{spg:fig:ana:lowerbound:sc_situation}
	\end{figure}
	
	From the rightmost node in $A$, there exists a path to a closest node in $V_S$, which is stretched by all the nodes between $A$ and $I$ that do not belong to $P_A$.
	Denote the set of these nodes by $O_L$.
	Let $x$ be the right neighbor of the rightmost node in $A$.
	If $x \notin V_S$, there exists a path from $x$ to a closest node in $V_S$ (since the terminals of the component $P_x$ that $x$ belongs to must definitely be in $V_S$), which is stretched by, among others, the nodes from $P_A$ that lie between $x$ and $I$.
	Denote the set of these nodes by $O_A$.
	If, however, $x \in V_S$, we have $O_A = \emptyset$ (and the following holds accordingly).
	All in all, since we identified the stretchings of disjoint paths, we have an additional stretching of $|O_A| + |O_L|$ induced by the nodes in the part left of $I$.
	The right part (starting from $B$) is analog, so we also have an additional stretching of $|O_B| + |O_R|$, where $O_B$ and $O_R$ are defined analogous to $O_A$ and $O_L$, respectively.

	Summing up the \vaddcosts we collected so far, yields $|I|-1 + |O_A| + |O_L| + |O_B| + |O_R|$.
	Since $I \cup O_A \cup O_L \cup O_B \cup O_R$ is the set of all nodes that are in $S$ but not in $A$ or $B$, this sum is equal to $|S| - 1 - |A| - |B|$.
	It holds that $|S| = \sum_{i=1}^{m}{|P_i^\mt|} + 1$.	
	Thus, we have:
	\[\VCOST_{\ARROPT}(S) \geq \sum_{i=1}^{m}{|P_i^\mt|} - |A| - |B|. \]
	Next, we apply some upper bounds for the sizes of $A$ and $B$ in order to rewrite this inequation as in the original claim of the lemma.

	\Wlog assume that $|A| \geq |B|$.
	This yields $|B| \leq \Delta_S$.
	Furthermore, it holds that $|A| \leq |P_A^\mt| - \Delta_{P_A}$.
	To understand this, assume for contradiction $|A| > |P_A^\mt| - \Delta_{P_A}$.
	Then, less than $|P_A| - (|P_A^\mt| - \Delta_{P_A}) = \Delta_{P_A} + 1$ nodes of $P_A$ would have been placed to the right of $A$ in $\ROPT{P_A}$.
	Since the terminal nodes of $P_A$ are in $V_S$ and thus lie to the right of $A$, this would imply that less than $\Delta_{P_A} - 1$ nodes from $P_A$ lie to the right of the rightmost terminal node of $P_A$ in $\ROPT{P_A}$.
	However, this would contradict to the minimality of $\Delta_{P_A}$.

	Recall that we assumed that the leftmost and rightmost nodes are not in $V_S$.
	This is due to the fact that if the leftmost node is in $V_S$, the sequence $A$ does not exist and we cannot define $O_A$ and $O_L$ as above.
	The analogous is true for $O_B$ and $O_R$ in case that the rightmost node is in $V_S$.
	However, with a case distinction, we can show that the above stated bounds hold in any case.
	First, assume that the rightmost node is in $V_S$, but the leftmost node is not.
	In this case, $I \cup O_A \cup O_L$ contains all nodes that are in $S$ but not in $A$.
	Thus, $|I| - 1 + |O_A| + |O_L|$ is exactly $|S| - 1 - |A|$.
	For $|B| := 0$, we can also write this as $|S| - 1 - |A| - |B|$.
	Note that $|B| \leq \Delta_S$ trivially holds in this case.
	Second, assume that the leftmost node is in $V_S$.
	This means that $A$ is empty, which implies (since we assumed, without loss of generality, $|A| \geq |B|$) that $B$ is empty, too.
	In this case, $|I| - 1 = |S| - 1 = |S| - 1 - |A| - |B|$ if we set $|A| := |B| := 0$.
	Furthermore, it is true that $|A| \leq |P_A^\mt| - \Delta_{P_A}$ for an arbitrary component $P_A \in S$.
	Thus, regardless of which case we are in, we can continue:
	\begin{align*}
	  \VCOST_{\ARROPT}(S) &\geq \sum_{P_i \in S}{|P_i^\mt|} - |A| - |B|\\
						 &\geq \sum_{P_i \in S}{|P_i^\mt|} - (|P_A^\mt| - \Delta_{P_A}) - \Delta_S\\
						 &= \sum_{P_i \in S \setminus \{P_A\}}{|P_i^\mt|} + \Delta_{P_A} - \Delta_S\\
						 &\geq \frac{1}{2}\sum_{P_i \in S \setminus \{P_A\}}{|P_i^\mt|} + \sum_{P_i \in S \setminus \{P_A\}}{\Delta_{P_i}} + \Delta_{P_A} - \Delta_S\\
						 &\geq \frac{1}{2}\left(\sum_{i=1}^m{|P_i^\mt|} - \max_i{|P_i^\mt|}\right) + \sum_{i=1}^m{\Delta_{P_i}} - \Delta_S\\
	\end{align*}
	in which we use that $\frac{1}{2}{|P_i^\mt|} \geq \Delta_{P_i}$ for any $i$ (c.f. Lemma~\ref{spg:lem:deltabound}) and that $|P_A^\mt| \leq \max_i{|P_i^\mt|}$.
\end{proof}

\optpc*
\begin{proof}
 Let $P$ be a parallel component with $m \geq 2$ child components $S_1, \dots, S_m$ in a \spg $G$ and let $\ARROPT$ be an optimal linear arrangement of $G$.
	Note that for parallel components, the \vaddcost is equal to the \addcost.
	Thus, it suffices to provide a lower bound on the \addcost of $P$.
	Similar to the proof of Lemma~\ref{spg:lem:optsc}, this is done by determining a lower bound on the sum of additional stretchings of all edges in $P$.
	
	In the following, we consider the arrangement $\ARROPT$ restricted to $P$. 
	Denote the terminals of $P$ by $u$ and $v$.
	Without loss of generality, assume that $u$ is placed to the left of $v$.
	Let $a$ be the leftmost node in $\ROPT{P}$ and $b$ be the rightmost node in $\ROPT{P}$.
	For now, assume $a \neq u$ and $b \neq v$ (we will handle the other cases later on).
	Denote by $S_A$ the (series) child component of $P$ that $a$ is from and denote by $S_B$ the (series) child component of $P$ that $b$ is from.
	Let $A$ be the maximal sequence of consecutive nodes in $\ROPT{P}$ that starts with $a$ and only contains nodes that belong to $S_A$ but do not equal $u$.
	Similarly, define a rightmost sequence $B$ that is the analogon to $A$ (i.e., a maximal sequence of consecutive nodes in $\ROPT{P}$ that start with $b$ and whose nodes belong to $S_B$ only, but $v$ must not belong to $B$).
	It is possible that $S_A = S_B$, which does not pose a problem in the following proof.
	The situation is depicted in Figure~\ref{spg:fig:ana:lowerbound:pc_situation}.
	\begin{figure}[htb]\centering
			\includegraphics[width=.5\linewidth]{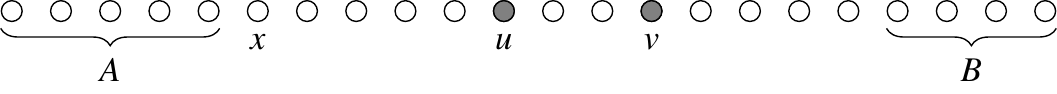}
		\caption{Example of $\ARROPT$ restricted to $P$ in the proof of Lemma~\ref{spg:lem:optpc}.}
		\label{spg:fig:ana:lowerbound:pc_situation}
	\end{figure}
	
	We now consider the nodes to the left of $v$ and to the right of $v$ individually.
	For any node whose position is left of $v$, except for $u$ and the nodes belonging to $A$, we show that there is an additional stretching of at least one due to this node.
	The analogous holds for the nodes whose positions are right of $v$, except for those belonging to $B$.
	This yields an additional stretching of at least $|P| - 2 - |A| - |B|$, which we will use to establish the lower bound in the claim of the lemma.

	Since $S_A$ is a series component with a path from $u$ to $v$, there is a path from the rightmost node of $A$ to $v$ consisting of edges from $S_A$ only.
	Let $E_L$ be the set of edges on this path.
	The additional stretching of any edge $e \in E_L$ is at least the number of nodes that lie between $A$ and $v$ and that are from different child components of $P$ (other than $S_A$), since the edges in $E_L$ span these nodes.
	Therefore, the sum of all additional stretchings of the edges in $E_L$ is at least the number of nodes in $\ROPT{P}$ that lie to the left of $v$ and do not belong to $S_A$.
	We denote this number by $|O_L|$.

	Now we consider the node $x$ that lies to the right of the rightmost node of $A$ in $\ROPT{P}$.
	By definition of $A$, either $x$ does not belong to $S_A$ or $x=a$.
	In the first case, denote the series component that $x$ belongs to by $S_x$.
	Otherwise, let $S_x$ be an arbitrary child component of $P$ such that $S_x \neq S_A$.
	There also exists a path from $x$ to $v$ consisting of edges from $S_x$ only.
	Let $E_x$ be the set of edges on this path.
	Since all the nodes of $S_A$ that lie between $x$ and $v$ (including $x$) contribute to the additional stretchings of the edges in $E_x$, the sum of additional stretchings due to the edges in $E_x$ is at least the number of these nodes, denoted by $|O_A|$.

	Considering only the section from $A$ to $v$, we have identified a sum of additional stretchings of at least $|O_L| + |O_A|$.
	Notice that only nodes that lie to the left of $v$ are counted for this argument.	
	Analogously, we can argue about the part from $B$ to $v$ and find a sum of additional stretchings of at least $|O_R| + |O_B|$, where $|O_R|$ is the number of nodes between $v$ and $B$ that do not belong to $S_B$, and $|O_B|$ is the number of nodes between $v$ and $B$ that do belong to $S_B$.
	Here, we only need to count nodes that lie to the right of $v$.
	In total, since $O_L \cup O_R \cup O_A \cup O_B$ is the set of all nodes in $P$ except for the two terminal nodes and those in $A$ or $B$, the sum of additional stretchings it at least the sum of the sizes of all (series) child components (not counting the two terminal nodes) minus the size of $A$ and the size of $B$.
	That is:
	      \[ \VCOST_{\ARROPT}(P) \geq \sum_{i=1}^{m}{|S_i^\mst|} - |A| - |B|. \]
	\Wlog, assume $|A| \geq |B|$.
	By definition of $A$ and $B$ (which must contain neither $u$ nor $v$), this yields $\Delta_P \geq |B|$.
	Furthermore, $|A| \leq |S_A^\mst| - \Delta_{S_A}$.
	To understand this, assume for contradiction $|A| > |S_A^\mst| - \Delta_{S_A}$.
	Then, less than $|S_A| - (|S_A^\mst| - \Delta_{S_A}) = \Delta_{S_A} + 2$ nodes of $S_A$ would have been placed to the right of $A$ in $\ROPT{S_A}$.
	Since $u$ and $v$ lie to the right of $A$ and belong to $S_A$, this would imply that less than $\Delta_{S_A}$ nodes from $S_A$ lie to the right of the rightmost node among $u$ and $v$ in $\ROPT{S_A}$.
	However, this would contradict to the minimality of $\Delta_{S_A}$.

	Recall that we assumed $a \neq u$ and $b \neq v$.
	We now handle the remaining cases.
	First of all, consider $a \neq u$ and $b = v$.
	In this case, the sequence $B$ does not exist and we can define $O_B := O_R := \emptyset$.
	Then, $O_L \cup O_A$ is the set of all nodes in $P$ except for the two terminal nodes and those in $A$.
	Thus, if we define $|B| := 0$, $\VCOST_{\ARROPT}(P) \geq \sum_{i=1}^{m}{|S_i^\mst|} - |A| - |B|$ holds, too. 
	Trivially, also $|B| \leq \Delta_P$ holds.	
		
	In the case $a = u$, the sequence $A$ is empty.
	Since we assumed, without loss of generality, $|A| \geq |B|$, this case can only hold if $b = v$, too.
	Let $S_A$ and $S_x$ be two arbitrary distinct child components of $P$.
	There exists a path from $u$ to $v$ through each of the two components.
	Denote the edges on this path by $E_L$ and $E_x$.
	Similar to the proof in the case $a \neq u$, the edges in $E_L$ are stretched by all nodes from $P$, except for $u$, $v$, and those in $S_A$, and the edges in $E_x$ are stretched by all nodes from $S_A$.
	All in all, we do not need to count any node twice to find a stretching of $\VCOST_{\ARROPT}(P) \geq \sum_{i=1}^{m}{|S_i^\mst|}$ in this case, too.
	To proceed as in the other cases, we define $|A| := |B| := 0$, such that $\VCOST_{\ARROPT}(P) \geq \sum_{i=1}^{m}{|S_i^\mst|} - |A| - |B|$ and $|B| \leq \Delta_P$ holds.
	
	All in all, we get:
	\begin{align*}
	  \VCOST_{\ARROPT}(P) &\geq \sum_{S_i \in P}{|S_i^\mst|} - |A| - |B| \\
						 &\geq \sum_{S_i \in P}{|S_i^\mst|} - |S_A^\mst| + \Delta_{S_A} - \Delta_P\\
						 &= \sum_{S_i \in P \setminus\{S_A\}}{|S_i^\mst|} + \Delta_{S_A} - \Delta_P\\
						 &\geq \frac{1}{2}\sum_{S_i \in P \setminus\{S_A\}}{|S_i^\mst|} + \sum_{S_i \in P \setminus\{S_A\}}{\Delta_{S_i}} + \Delta_{S_A} - \Delta_P\\
						 &\geq \frac{1}{2}\left(\sum_{i=1}^m{|S_i^\mst|} - \max_i{|S_i^\mst|}\right) + \sum_{i=1}^m{\Delta_{S_i}} - \Delta_P
	\end{align*}
	In the fourth step we use $\frac{1}{2}|S_i^\mst| \geq \Delta_{S_i}$, which is due to Lemma~\ref{spg:lem:deltabound}.
	The fifth step holds due to $|S_A| \leq \max_i{|S_i|}$.
\end{proof}

\opttotal*
\begin{proof}
	By Observation~\ref{spg:obs:totalcost}, $COST_{\ARROPT}(G)$ is the sum of all \addcosts of any series or parallel component $C$, i.e., $COST_{\ARROPT}(G) = \sum_{C\in\mathcal{C}}{\ACOST_{\ARROPT}(C)}$, where $\mathcal{C}$ is the set of all series or parallel components of $G$.
	Recall the relationship $3\cdot \sum_{C\in\mathcal{C}}{\ACOST_{\ARROPT}(C)} \geq \sum_{C \in \mathcal{C}}{\VCOST_{\ARROPT}(C)}$ from Lemma~\ref{spg:lem:sumofvaddcostatmostthreetimessumofaddcost}.
	This gives us:
	\begin{align*}
	  &6\cdot COST_{\ARROPT}(G) \geq 2\cdot\left( \sum_{C \in \mathcal{C}}{\VCOST_{\ARROPT}(C)} \right) \\
			&= 2\cdot\left( \sum_{L \in L_G}{\VCOST_{\ARROPT}(L)} + \sum_{P \in P_G}{\VCOST_{\ARROPT}(P)} + \sum_{S \in S_G}{\VCOST_{\ARROPT}(S)} \right)\\
			&\geq 2\cdot\left( \sum_{L \in L_G}{\left(|L|-1 + \Delta_L\right)} \right.\\
				&\quad\qquad + \left.\sum_{P \in P_G}{\bigl(\frac{1}{2}\bigl(\sum_{S_i \in P}{|S_i^\mst|} - \max_{S_i \in P}{|S_i^\mst|}\bigr) + \sum_{S_i \in P}{\Delta_{S_i}} - \Delta_P\bigr)} \right.\\
				&\quad\qquad + \left.\sum_{S \in S_G}{\bigl(\frac{1}{2}\bigl(\sum_{P_i \in S}{|P_i^\mt|} - \max_{P_i \in S}{|P_i^\mt|}\bigr) + \sum_{P_i \in S}{\Delta_{P_i}} - \Delta_S\bigr)}\right)
	\end{align*}
	in which we apply Lemma~\ref{spg:lem:optsns}, Lemma~\ref{spg:lem:optsc}, and Lemma~\ref{spg:lem:optpc}.
	We now take a closer look at the relationship between the different sets.
	Denote by $O$ the outmost (series or parallel) component.
	For each \sns $L \in L_G \setminus\{O\}$, there is a series or parallel component which has $L$ as a child.
	Furthermore, for each series component $S \in S_G\setminus\{O\}$, there is a parallel component which has $S$ as a child component.
	Similarly, for each parallel component $P \in P_G\setminus\{O\}$, there is a series component which has $P$ as a child component.	
	Since each (parallel or series) component adds a value of $\Delta_{C_i}$ for each of its child components $C_i$ to the cost, the subtracted $\Delta_{C_i}$ values for all components $C_i$ (except for $O$) are canceled out by these.
	All that remains is a subtrahend of $\Delta_O$ for the outmost component.
	Note that $\Delta_O \leq |V|$ and $COST_{\ARROPT}(G) \geq |V|$, which gives us $COST_{\ARROPT}(G) \geq \Delta_O$.
	Thus, the seventh addition of $COST_{\ARROPT}(G)$ in the original claim cancels out the $\Delta_O$ to be subtracted, which concludes the proof.
\end{proof}

\algtotal*
In order to prove Corollary~\ref{spg:cor:algtotal}, we first bound the \addcost of \snss, parallel components and series components individually.
This is given by the following three lemmata.

We start with a bound on the \addcost of \snss:
\begin{restatable}{lemma}{algsns}\label{spg:lem:algsns}
 For any \sns $L$ in an arrangement $\ARRALG$ computed by the \spgalg, it holds:
	\[ \ACOST_{\ARR_{ALG}}(L) \leq 2 \cdot (|L|-1). \]
\end{restatable}
\begin{proof}
 Consider the way the \spgalg arranges a simple node sequence (c.f. Subsection~\ref{spg:subsec:alg:sns}).
  Observe that the length of each edge in $\RALG{L}$ is at most two.
	Since there are exactly $|L|-1$ edges in a \sns, this implies the claim.
\end{proof}

Next we upper bound the \addcost of parallel components as arranged by the \spgalg:
\begin{restatable}{lemma}{algpc}\label{spg:lem:algpc}
 For any parallel component $P$ in a \spg $G$ with $m \geq 2$ child components $S_1, \dots, S_m$, in an arrangement $\ARRALG$ computed by the \spgalg, it holds:
	\[ \ACOST_{\ARR_{ALG}}(P) \leq 2\cdot D^2 \cdot \left(\sum_{i=1}^{m}{|S_i^\mst|} - \max_i{|S_i^\mst|}\right). \]
\end{restatable}
\begin{proof}
	Let $P$ be a parallel component in a \spg $G$ with $m \geq 2$ child components $S_1, \dots, S_m$, in an arrangement $\ARRALG$ computed by the \spgalg.

 Recall that for parallel components, the \addcost arise due to the increase of distances between connected nodes in the current recursion level in comparison to the previous recursion level.
	From the perspective of a single child component $S_i$, all that changes in the current recursion step is that the two terminal nodes of $S_i$ are moved by a certain distance to the left (c.f. the algorithm description in Subsection~\ref{spg:subsec:alg:pc}).
	Therefore, only the lengths of edges whose one endpoint is a terminal node can increase at all (since the remaining edges keep their previous lengths).
	For the whole component, this can be at most $2D$ edges (since each of the terminals is connected with at most $D$ other nodes from $S_i$).

	The extent to which these edge lengths increase depends on the position of $S_i$ in $\ARRALG$:
	For the leftmost component (\Wlogs, let this be $S_1$), this is zero.
	For the second component (\Wlogs, let this be $S_2$, accordingly), it is exactly $|S_1^\mst|$.
	The third component (\Wlogs, let this be $S_3$) has an increase of $|S_1^\mst| + |S_2^\mst|$ and so on.
	For the last component (which, according to the description of the \spgalg, is the biggest component and denoted by $S_m$), the extent to which the up to $2D$ edges to the terminal nodes are stretched is $\sum_{i=1}^{m-1}{|S_i^\mst|} = \sum_{i=1}^m{|S_i^\mst|} - \max_i{|S_i^\mst|}$.

	All in all, we have that for any of the $m \leq D$ child components, at most $2D$ edges are stretched by an extent of at most $\sum_{i=1}^m{|S_i^\mst|} - \max_i{|S_i^\mst|}$, which completes the proof.	
\end{proof}

The last missing piece is the upper bound of the \addcost of series components in arrangements of the \spgalg:
\begin{restatable}{lemma}{algsc}\label{spg:lem:algsc}
 For any series component $S$ in a \spg $G$ with $m \geq 2$ child components $P_1, \dots, P_m$, in an arrangement $\ARRALG$ computed by the \spgalg, it holds:
	\[ \ACOST_{\ARR_{ALG}}(S) \leq 2\cdot D \cdot \left(\sum_{i=1}^m{|P_i^\mt|} - \max_i{|P_i^\mt|}\right). \]
\end{restatable}
\begin{proof}
		Let $S$ be a series component in a \spg $G$ with $m \geq 2$ child components $P_1, \dots, P_m$, in an arrangement $\ARRALG$ computed by the \spgalg.

	Again, the \addcost arise due to nodes that have a greater distance to other nodes when compared to the previous recursion level.
	We now identify the cases in which a distance is increased.
	First of all, observe that for a fixed child component $P_i$, any non-terminal node always keeps its distance to any other non-terminal node.
	Thus, we only need to identify the cases in which a terminal node increases its distance to other nodes and to what extent.
	Summing up these values yields the desired \addcost then.

	For the purpose of finding out the values, it is helpful to consult Figure~\ref{spg:fig:alg:sc}.
	Here, the dotted nodes are those that change their position (compared to the previous recursion level) and the lengths of the arrows (i.e., the number of nodes that they span) indicate the distance by which they are moved.
	We now just determine the lengths of the arrows for each of the dotted nodes one after another.
	Since each such node can be connected with at most $D$ other nodes from its component, multiplying the length with $D$ yields the \addcost caused by this node movement.
	We consider the case $1 < a < m$ first, with $a$ defined as in the description of the \spgalg (see Subsection~\ref{spg:subsec:alg:sc}).

	The first dotted node is the source of $P_1$ (which is also the source of $S$).
	Its distance to other nodes from $P_1$ increases by at most one.

	The second dotted node is the sink of $P_1$ (which is also the source of $P_2$).
	It is moved to the right end of $P_1$ at which it is shifted by a distance of $|P_1^\mt| - 1$.
	Analogously, the next dotted node (the sink of $P_2$) increases its distance to other nodes by $|P_2^\mt| - 1$.
	The same holds for the sink of each child component $P_j$ with $3 \leq j \leq a-2$: Each one is shifted by $|P_j^\mt| - 1$.
	All in all, for the sinks of the child components from $P_2$ to $P_{a-2}$, we have a total increase in distances of up to $\sum_{i=1}^{a-2}{(|P_i^\mt|-1)}$.

	The sink of $P_{a-1}$ is shifted by a greater distance:
	It passes all other nodes of $P_{a-1}$ and all nodes of $P_j$ for $m \geq j \geq a+1$. 
	That is, in total its increase in distances is at most $|P_{a-1}^\mt| - 1 + \sum_{i=a+1}^{m}{|P_i^\mt|}$ (recall that we need to add the non-dashed nodes only).

	The next dashed node we have not considered so far is the sink of $P_m$.
	It is moved to the second position (from left to right) and thereby passes all other nodes of $P_m$ (except for the source), and all nodes of $P_j$ for $a-1 \geq j \geq 1$.
	This makes up a total distance of $|P_m^\mt| - 1 + \sum_{i=1}^{a-1}{|P_i^\mt|} - 1$ (the last minus one is due to the source of $P_1$ whose movement has already been counted).

	The following dotted nodes, starting with the sink of $P_{m-1}$ and ending with the sink of $P_{a+1}$ move by the size of their component (without the two terminals) each, which makes a total distance of $\sum_{i=a+1}^{m-1}{(|P_i^\mt|-1)}$.

	Last, the sink of $P_a$ is moved by one.

	Summing over all these values yields the following upper bound on the \addcost:
\begin{align*}
	\ACOST_{\ARR_{ALG}}(S) \leq& 1 + \left(\sum_{i=1}^{a-2}{(|P_i^\mt|-1)}\right) + \left(|P_{a-1}^\mt| - 1 + \sum_{i=a+1}^{m}{|P_i^\mt|} \right)\\
	&
	+ \left( |P_m^\mt| - 1 + \sum_{i=1}^{a-1}{|P_i^\mt|} - 1 \right) + \left( \sum_{i=a+1}^{m-1}{(|P_i^\mt|-1)} \right) + 1\\
	=& 2\cdot\left(\sum_{i=1}^m{|P_i^\mt|}-\max_i{|P_i^\mt|}\right) + 2 - m
	\leq 2\cdot\left(\sum_{i=1}^m{|P_i^\mt|}-\max_i{|P_i^\mt|}\right)
	\end{align*}
	in which we use that $m \geq 2$.
	The computed value is the sum of all distances that each terminal node has moved and thus possibly increased the distance to at most $D$ other nodes in its own component.
	Therefore, this completes the proof in this case.

For the case $a=1$, the procedure is similar to the previous case.
Actually, the increase in distance caused by the movement of the dotted nodes from the sink of $P_{m-1}$ up to the sink of $P_2$ is exactly the same.
Its value is $\sum_{i=2}^{m-1}{(|P_i^\mt|-1)}$.
Additionally, we have to take into account the movement of the sink of $P_m$ (by a distance of $|P_m^\mt|-1$) and of the source of $P_1$ (which passes the nodes of all other child components, which are $\sum_{i=2}^{m}{|P_i^\mt|} + 1$ in total).
If we assume that the source of $P_1$ is moved first, the sink of $P_1$ does not need to be moved any more.
All in all, we have a total possible increase of distances (and, accordingly, an \addcost) of at most 
\[ \ACOST_{\ARR_{ALG}}(P) \leq 2\cdot\sum_{i=2}^{m}{|P_i^\mt|} + \sum_{i=2}^{m}{(-1)} + 1 \leq 2\cdot\left(\sum_{i=1}^{m}{|P_i^\mt|} - \max_i{|P_i^\mt|}\right), \]
which, following the argument from the previous case, completes the proof in this case, too.

For the case $a=m$, the movement of all of the dotted nodes up to the sink of $P_{a-2} = P_{m-2}$ is exactly the same.
For these we have a total increase in distances of $1 + \sum_{i=1}^{m-2}{(|P_i^\mt|-1)}$.
Next, we have the sink of $P_{m-1}$ which is shifted by a distance of $|P_{m-1}^\mt| - 1$.
Last, the sink of $P_m$ is shifted by a distance of $1 + \sum_{i=1}^{m-1}{|P_i^\mt|} - 1$ (the last minus one originates from the assumption that we shift the source of $P_1$ first and thus the sink of $P_m$ does not need to pass it any more).
All in all, we have a total possible increase of distances (and, accordingly, an \addcost) of at most
\[ \ACOST_{\ARR_{ALG}}(P) \leq 2\cdot \sum_{i=1}^{m-1}{|P_i^\mt|} = 2\cdot \left( \sum_{i=1}^{m}{|P_i^\mt|} - \max_i{|P_i^\mt|} \right). \]
With the same argument as in the previous two cases, this completes the proof in this case, too.
\end{proof}

These three lemmata (together with Observation~\ref{spg:obs:totalcost}) enable us to prove Corollary~\ref{spg:cor:algtotal}:
\algtotal*
\begin{proof}
Let $\mathcal{C}$ be the set of all components of $G$.
Recall that $COST_{\ARRALG}(G) = \sum_{C\in\mathcal{C}}{\ACOST_{\ARR_{ALG}}(C)}$ (by Observation~\ref{spg:obs:totalcost}).
	Applying Lemma~\ref{spg:lem:algsns}, Lemma~\ref{spg:lem:algpc}, and Lemma~\ref{spg:lem:algsc} yields:
	\begin{align*}
	COST_{\ARRALG}(G) &= \sum_{L \in L_G}{\ACOST_{\ARR_{ALG}}(L)} + \sum_{P \in P_G}{\ACOST_{\ARR_{ALG}}(P)} + \sum_{S \in S_G}{\ACOST_{\ARR_{ALG}}(S)}\\
			&\leq \sum_{L \in L_G}{2\cdot(|L|-1)}\\
				&\quad + \sum_{P \in P_G}{2\cdot D^2\cdot\left(\sum_{S_i \in P}{|S_i^\mst|} - \max_{S_i \in P}{|S_i^\mst|}\right)}\\
				&\quad + \sum_{S \in S_G}{2\cdot D\cdot\left(\sum_{P_i \in S}{|P_i^\mt|} - \max_{P_i \in S}{|P_i^\mt|}\right)},
	\end{align*}
which completes the proof.
\end{proof}

\thmapproxratio*
\begin{proof}
 Let $G$ be a \spg, let $\ARRALG$ be the linear arrangement of $G$ computed by the \spgalg, and let $\ARROPT$ be an optimal linear arrangement of $G$.
	Corollary~\ref{spg:cor:opttotal} implies:
	\begin{equation*}
	\begin{split}
		14 \cdot D^2 \cdot COST_{\ARROPT}(G) \geq& 2D^2\left( \sum_{L \in L_G}{2(|L|-1)} + \sum_{P \in P_G}{\left(\sum_{S_i \in P}{|S_i^\mst|} - \max_{S_i \in P}{|S_i^\mst|}\right)} \right.\\
		  &\qquad \left. + \sum_{S \in S_G}{\left(\sum_{P_i \in S}{|P_i^\mt|} - \max_{P_i \in S}{|P_i^\mt|}\right)} \right).
	\end{split}
 \end{equation*}

	On the other hand, from Corollary~\ref{spg:cor:algtotal}, we can derive:
	\begin{align*}
	COST_{\ARRALG}(G) \leq& 2 \cdot D^2 \cdot \left( \sum_{L \in L_G}{(|L|-1)}\right.
			 \left. + \sum_{P \in P_G}{\left(\sum_{S_i \in P}{|S_i^\mst|} - \max_{S_i \in P}{|S_i^\mst|}\right)}\right.\\
			&\qquad\qquad\left.+ \sum_{S \in S_G}{\left(\sum_{P_i \in S}{|P_i^\mt|} - \max_{P_i \in S}{|P_i^\mt|}\right)}\right)
	\end{align*}

	Together, these two equations yield the claim and thus complete the proof of this theorem.
\end{proof}

\runtime*
\begin{proof}
For this proof, we assume that the \mspt $T$ used to define the series and parallel components of a graph $G=(V,E)$ is given as an input to the algorithm.
If $T$ is not given, the algorithm described in Appendix~\ref{sec:appendix:scs} can be used to compute it in time $O(|E|\log{|E|})$ (thus the runtime bound is slightly higher in that case).

Let $s$ be the sum of the sizes (\wrt the number of nodes in it) of all \snss and let $c$ be the total number of (series or parallel) components in $G$.
We now determine upper bounds on $s$ and $c$ and then bound the runtime depending on their values.

First of all, notice that each edge in $G$ can be part of at most one \sns.
Furthermore, any \sns has at most twice as many nodes as edges.
Thus, $s$ can be at most $2|E|$ in total.

Second, note that the number of edges in any result of a series or parallel composition operation $OP$ is strictly greater than the number of edges in each input to $OP$.
Thus, $c$ is upper bounded by $|E|$.

It is possible to implement the \spgalg in the following way:
The algorithm traverses $T$ from bottom to top and constructs an internal representation of the \spg as well as its linear arrangement during this traversal.
At each node $u$ of $T$, the arrangement of the subgraph corresponding to $u$ is determined (basically by ordering the child components of $u$ in the right way). 
This way, it is obvious that the runtime of the algorithm is upper bounded by the number of nodes of $T$ (which is $c$) plus the time spent at each edge of $T$.
Therefore, we now determine the latter.

At the lowest recursion level, the algorithm arranges the \snss.
For each \sns $S$, determining the order of the nodes in $S$ is possible in time linear in $|S|$.
Thus, for arranging all \snss, we need time linear in $s$ only.

Throughout the higher recursion levels, the algorithm determines the order of the child components of all series or parallel components that are not \snss.
For any such component $C$, this is possible in time linear in the number of child components of $C$.
Since each component can be a child component of at most one other component, this implies that the higher recursion levels require time linear in $c$ in total.

All in all, the number of nodes in $T$ is upper bounded by $O(|E|)$ and the total time spent at each node of $T$ is upper bounded by $O(s+c) = O(|E|)$.
Thus, the algorithm has a total runtime of $O(|E|)$.
\end{proof}

\section{Additional figures}\label{sec:appendix:figures}
\begin{figure}[htb]\centering
	\subfloat[width=.49\linewidth][$SNS(6)$.]{\centering
		\includegraphics[scale=1]{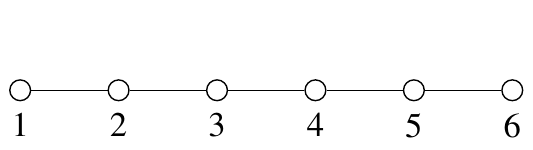}	
	}
	\hfill
	\subfloat[width=.49\linewidth][The resulting arrangement of $SNS(6)$.]{\centering
		\includegraphics[scale=1]{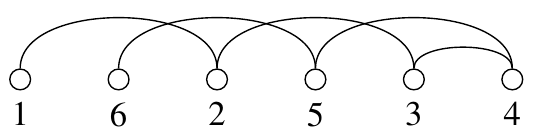}
	}

	\caption{\spgalg arrangement of a \sns.}
		\label{spg:fig:sns_example}
\end{figure}

\begin{figure}[htb]\centering
	\subfloat[width=.42\linewidth][$P$ in the original graph.]{
		\centering
		\includegraphics[scale=1.0]{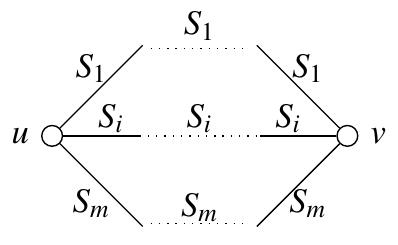}
	}
	\hspace{60pt}
	\subfloat[width=.42\linewidth][The resulting arrangement of $P$. Dashed nodes represent former positions of $u$ and $v$.]{
		\includegraphics[width=.42\linewidth]{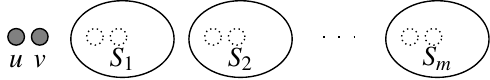}
	}
	\caption{The order in which the \spgalg arranges a parallel component $P$ with source $u$ and sink $v$ consisting of $m$ series child components $S_1, \dots, S_m$. Note that we assume $S_m$ to be a biggest child component.}
		\label{spg:fig:pc_arrangement_example}
\end{figure}

\begin{figure}[htb]\centering
	\subfloat[width=\linewidth][$1 < a < m$]{\centering
		\includegraphics[scale=1]{spgdefseriescomponentorderingexampleA}
	}
	
	\subfloat[width=\linewidth][$a = 1$]{\centering
		\includegraphics[scale=1]{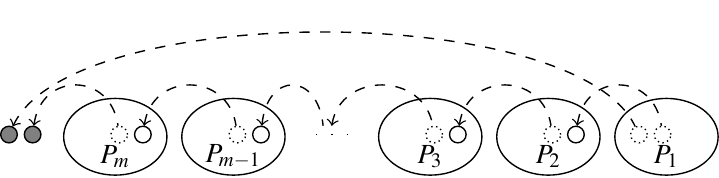}
	}
	
	\subfloat[width=\linewidth][$a = m$]{\centering
		\includegraphics[scale=1]{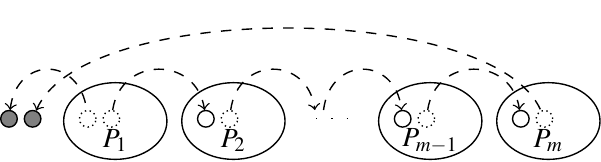}
	}
	\caption{The order in which the \spgalg arranges a series component consisting of $m$ child components. $P_a$ is a biggest child component.
	Dashed nodes indicate the position at which a node would be placed according to the previous recursion level.
	Dashed arrows indicate the change in position at the current recursion level.
	}
\end{figure}

\begin{figure}[tb]\centering
	\subfloat[width=.55\linewidth][Extract of an example graph with labels of some components.]{\centering
		\includegraphics[width=.55\linewidth]{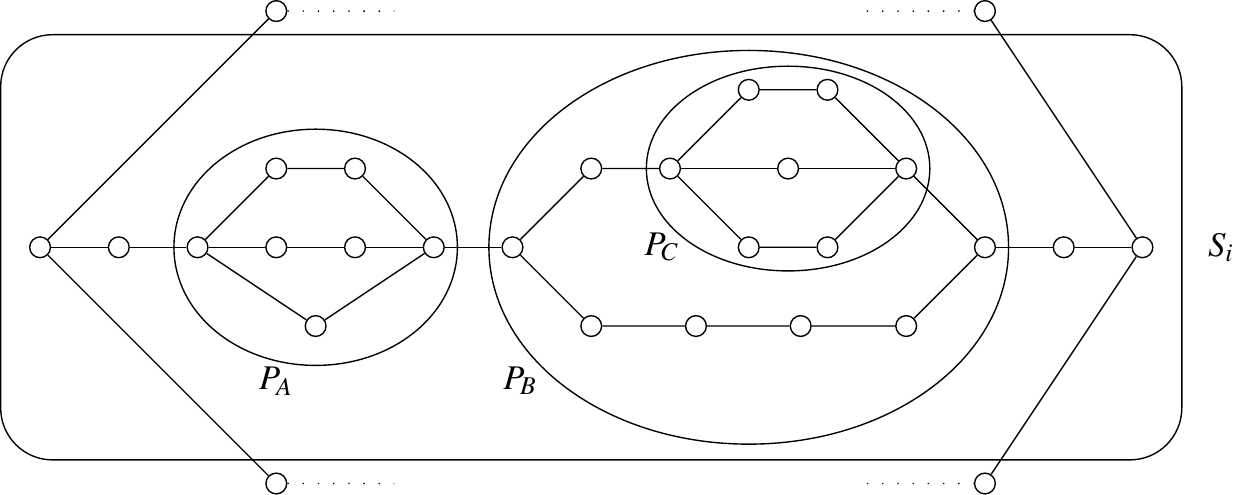}
	}
	
	\subfloat[width=.47\linewidth][Selection of one \apath and two \spaths for each of the two components $P_A$ and $P_C$.]{\centering
		\includegraphics[width=.47\linewidth]{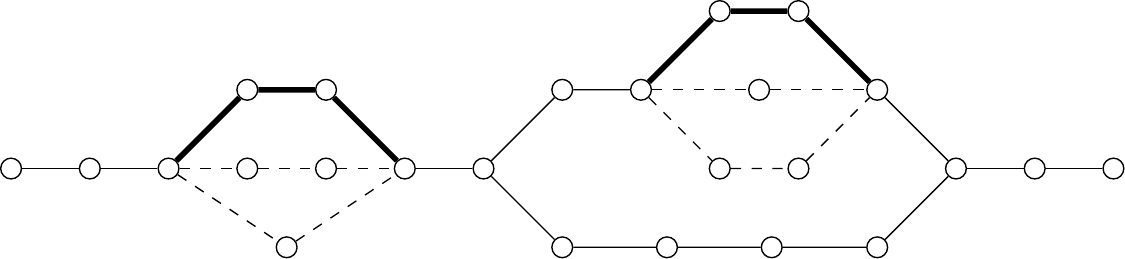}
	}
	\hfill
	\subfloat[width=.47\linewidth][Subsequent selection of an \apath and an \spath of $P_B$.]{
	  \label{spg:fig:sed_examplec}
	  \centering
		\includegraphics[width=.47\linewidth]{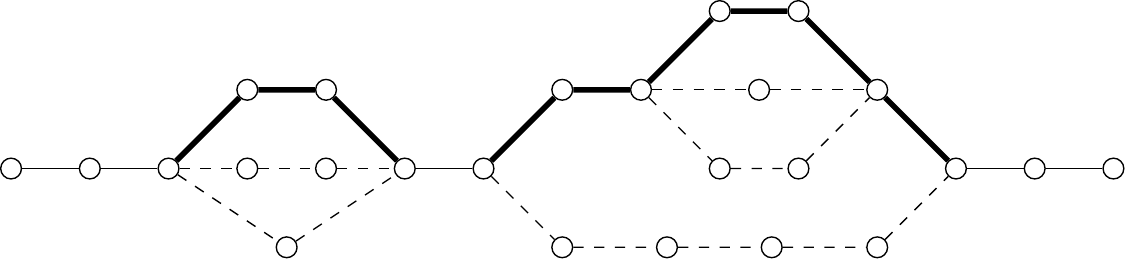}
		
	}
	\caption{Step-by-step illustration of an \sed.
					 After the first figure, only $S_i$ is shown.
					 Thick lines describe an \apath, whereas dashed lines describe \spaths.}
	\label{spg:fig:sed_example}
\end{figure}

\end{appendix}

\end{document}